\documentclass[reqno]{amsart}

\usepackage{amsmath,amssymb,amsthm,amsfonts, tikz}
\usepackage{color}
\usepackage{srcltx}
\usepackage{hyperref}

\theoremstyle{plain}
\newtheorem{prop}{Proposition}[section]
\newtheorem{theorem}[prop] {Theorem}

\theoremstyle{definition}
\newtheorem{lemma}[prop]{Lemma}

\newtheorem{cor}[prop]{Corollary}

\theoremstyle{remark}
\newtheorem*{remark}{Remark}
\newtheorem*{example}{Example}

\def \({\left(}
\def \){\right)}

\def \nn{\nonumber}

\def \bs{\begin{slide} }
\def \es{\end{slide} }

\def \bea{\begin{eqnarray}}
\def \eea{\end{eqnarray}}
\def \bes{\begin{eqnarray*}}
\def \ees{\end{eqnarray*}}
\def \nab{\nu^0_{a,b}}
\def \muxx{\mu_{xx}}
\def \muxy{\mu_{xy}}
\def \nuxy{\nu_{xy}}
\def \nur{\nu^0_{\R}}

\newcommand{\cI}{{\mathcal I}}

\newcommand{\N}{\mathbb{N}}
\newcommand{\R}{\mathbb{R}}
\newcommand{\Z}{\mathbb{Z}}

\renewcommand{\P}{\mathbb{P}}
\newcommand{\E}{\mathbb{E}}

\def\K{\mathbb{K}}

\newcommand{\dd}{\mathrm{d}} 
\newcommand{\eps}{\varepsilon}
\newcommand{\la}{\langle}
\newcommand{\ra}{\rangle}
\newcommand{\sgn}{\mathrm{sgn}} 

\newcommand{\vect}[1]{\boldsymbol{#1}}  
\newcommand{\be}{\begin{equation}}
\newcommand{\ee}{\end{equation}}

\DeclareMathOperator{\Tr}{Tr}
\DeclareMathOperator{\card}{card} 

\newcommand{\Hi}{\mathcal{H}} 

\title{Wigner crystallization in the quantum 1D jellium at all densities}
\author{S. Jansen} 
\address{Fakult{\"at} f{\"u}r Mathematik, Ruhr-Universit{\"a}t Bochum, 44780 Bochum, Germany} 
\email{sabinecjansen@gmail.com}
\author{P. Jung}
\address{University of Alabama, Birmingham, USA} 
\email{pjung@uab.edu}
\date{June 26, 2013}

\begin{document}

\maketitle

\begin{abstract}

The jellium is a model, introduced by Wigner (1934), for a gas of electrons moving in a uniform neutralizing background of positive charge. Wigner suggested that the repulsion between electrons might lead to a broken translational symmetry.
For classical one-dimensional systems this fact was proven by Kunz (1974), while in the quantum setting,
Brascamp and Lieb (1975) proved translation symmetry breaking  at low densities. Here, we prove translation symmetry breaking for the quantum one-dimensional jellium at all densities.\\

\noindent
\footnotesize{Keywords: Wigner crystal, jellium, quantum Coulomb system, translation symmetry breaking. 
}

\end{abstract}

\tableofcontents

\section{Introduction}
In a landmark paper \cite{wigner1934interaction}, Wigner introduced the jellium model for a gas of electrons and predicted that when the potential energy of the system overwhelms the kinetic energy,  the electrons would form a ``close-packed lattice configuration''.  We are interested in the one-dimensional quantum jellium (using the potential $-|x|$) which models uniformly charged parallel sheets which are able to move in the transverse direction.  Alternatively, one can consider this as a model of electrons inside a very thin insulated conducting wire.  In the case of wires which are not insulated (e.g., with interaction $1/r$), Schulz~\cite{schulz1993wigner} predicted that the crystallization could still persist in the weaker form of quasi long-range order. Deshpande and Bockrath~\cite{deshpande2008one} recently observed Wigner-crystal type behavior in experiments on carbon nanotubes.

Wigner crystallization was proven for the \emph{classical} one-dimensional jellium in \cite{kunz1974one, brascamp1975some} by showing periodicity of the one-point correlation functions at almost all densities and at all low electron densities (high spacing).  Later, \cite{aizenman1980structure, aizenman2001bounded} used ergodic-theoretic arguments to show crystallization at all densities.
This argumentation was extended to the classical quasi-one-dimensional\footnote{One infinite dimension crossed with a compact manifold.} jellium by \cite{aizenman2010symmetry}.

Wigner's original model was of course in the quantum setting, and  Kunz states that ``the quantum case (in one-dimension) might be rewarding for the following reasons: whereas a crystalline phase will certainly survive in the strong coupling regime, quantum effects become important in the weak coupling limit. But in this regime classically the particles already tend to be delocalized, so that it is not excluded that the uncertainty and the exclusion principle will keep the particles sufficiently far apart that they go into a gas phase. But evidently the presence or absence of such a transition remains to be proved.''\footnote{In the \emph{strong coupling} regime ($\beta /\rho \gg 1$), the potential energy dominates the entropy and entropic fluctuations should not destroy the periodicity. At weak coupling, this is less clear, especially since the quantum system has more fluctuations than the classical system (due to fluctuations of Brownian bridges in the Feynman-Kac picture).}
Only a year later \cite{brascamp1975some} proved, using the now well-known Brascamp-Lieb inequality for Gaussian measures,  crystallization for the one-dimensional quantum system at sufficiently low densities (again, through
the periodicity of the one-particle density).
In this work, we generalize and combine some arguments of \cite{kunz1974one} and \cite{aizenman1980structure} to prove crystallization for the quantum one-dimensional jellium at all densities.
The main tool we use is the Krein-Rutman theorem (a generalization of Perron-Frobenius) applied to a Ruelle transfer operator.


In the next section, we introduce the model and state our main results. In Section \ref{sec:tools}, we present some tools and lemmas that are used in the  proofs of the main results.  The final two sections are devoted to those proofs.  In the appendix, we present a connection between our model and a family of non-colliding Ornstein-Uhlenbeck bridges.
\section{Model and results}

In this section we formulate the model and results in the language of quantum statistical mechanics. The path integral formulation and associated probabilistic setup for systems of non-colliding Gaussian bridges are given in the next section.

Consider $N$ particles of charge $-1$ each, with positions $$x_1,\ldots, x_N \in [a,b] \subset \R.$$ The charges interact with each other through the one-dimensional Coulomb potential
$V(x):= - |x|$, $x\in \R$. Note the distributional  identity $- V''= 2 \delta_0$.  In addition, there is a neutralizing  background of homogeneous charge density $\rho =N/(b-a)$; the inverse density  $\lambda=\rho^{-1}$ denotes the typical spacing between the electrons.  The total potential energy of the system is
\begin{multline}\label{eq:potential}
    U(x_1,\ldots,x_N) := -\sum_{1 \leq j \leq k \leq N} |x_j- x_k| + \rho \sum_{j=1}^N \int_a^b |x_j - x| \dd x \\
        - \frac{\rho^2}{2} \int_a^b \int_a^b |x-x'| \dd x \dd x'.
\end{multline}

The Hilbert space $\Hi_N$ for $N$ fermions on the line $[a,b]$ is the space of square-integrable complex-valued functions $f\in L^2([a,b]^N)$ that are antisymmetric,
 i.e., $f(x_{\sigma(1)},\ldots, x_{\sigma(N)}) = \sgn (\sigma)\, f(x_1,\ldots,x_N)$ for all permutations $\sigma$ of $\{1,\ldots,N\}$.
The quantum-mechanical Hamiltonian is
\be
    H_N:= - \frac{1}{2}\sum_{j=1}^N \frac{\partial^2}{\partial x_j^2} + U(x_1,\ldots,x_N).
\ee
Here we use, as usual, the same letter $U$ for the function and for the associated multiplication operator, and we have chosen units such that $\hbar^2/m=1$.   We take Dirichlet boundary conditions, i.e., $H_N$ is the closure of the operator with domain $C_0^\infty((a,b)^N) \cap \Hi_N$. It is well-known that $H_N$ is self-adjoint (see \cite[Thm X.28]{reed1980methods2} and \eqref{eq:potential2} below).

Fix an inverse temperature $\beta>0$. The canonical partition function is
\be
    Z_N(\beta):= \Tr \exp( - \beta H_N).
\ee
In the next section we shall express $Z_N(\beta)$,  via the Feynman-Kac formula, as the expectation of a functional of Brownian bridges. Standard arguments show that $\exp(-\beta H_N)$ is an integral operator with continuous kernel $\exp(-\beta H_N)(x_1,\ldots,;\ldots,y_N)$ \cite{ginibre1965reduced, reed1980methods2}.

 For $n\in \{1,\ldots,N\}$, the \emph{$n$-point correlation function} or \emph{reduced density matrix}  \cite{ginibre1965reduced, bratteli1981operator}  is the  function of  $\vect{x}= (x_1,\ldots,x_n)$ and $\vect{y}=(y_1,\ldots,y_n)$, proportional to
\be
    \rho_n^N(\vect{x};\vect{y}) \propto
        \int_{[a,b]^{N-n}} e^{-\beta H_N} (\vect{x},\vect{x'};\vect{y},\vect{x'}) \dd \vect{x'}
\ee
with proportionality constant fixed by the condition
\be \label{eq:corr-normalization}
    \int_{[a,b]^n} \rho_n^N(\vect{x};\vect{x}) \dd x_1\cdots \dd x_n = N(N-1) \cdots (N-n+1).
\ee
The \emph{one-particle density} $\rho_1^N (x;x)$ (corresponding to $n=1$ and $x=y$) represents the expected number of particles in an infinitesimal neighborhood of $x$.
Eq.~\eqref{eq:corr-normalization} becomes $\int_a^b \rho_1^N(x;x) \dd x = N$ and expresses that there are $N$ particles in $[a,b]$.
The reduced density matrices inherit the antisymmetry, e.g.,
\be  \label{eq:corr-antisymmetry}
    \rho_2^N(x_1,x_2;y_1,y_2) = - \rho_2^N(x_1,x_2;y_2,y_1) = \rho_2^N(x_2,x_1;y_2,y_1).
\ee

We are now ready to state our results on the thermodynamic limit: we fix
the background charge density $\rho$ and inverse temperature $\beta$.  Recalling that \mbox{$N = \rho(b-a)$,} we let
 \be \label{eq:thermolim}
    \quad a\to - \infty, \quad b \to \infty \quad\text{(thus $N\to\infty$).}
\ee

Our first result concerns the asymptotics, under \eqref{eq:thermolim}, of the partition function $Z_N(\beta)$. Recall that the free energy
$$f(\beta,\rho):= - \lim \frac{1}{\beta N} \log Z_N(\beta),$$
was shown to exist in \cite{lieb1976thermodynamic} (their proof is with respect to the Coulomb potential $|x|^{-1}$, but as they state on p. 292, applies to $-|x|$ as well).

\begin{theorem}[Surface corrections] \label{thm:free-energy}
    The following limit along~\eqref{eq:thermolim} exists in $\R$:
    \be \nn
        \beta s(\beta,\rho):= - \lim \Bigr(\log Z_N(\beta) + N f(\beta,\rho) \Bigr).
    \ee
\end{theorem}

One can rewrite the above as  $$- \beta^{-1}\log Z_N(\beta) = N f(\beta,\rho) + s(\beta,\rho) + o(1)$$ in order to see that
the finite-size corrections to the free energy are bounded.
The free energy itself can be written, as we will see in Section \ref{sec:free energy}, as
\be\label{eq:free-energy-z0}
     f(\beta,\rho) = \frac{1}{12\rho} + \Bigl( \sqrt{\frac{\rho}{2}} + \frac{1}{\beta} \log ( 1- e^{- \beta \sqrt{2\rho}} )\Bigr) - \frac{1}{\beta}    \log z_0(\beta,\rho).
\ee
The first term is the minimum of the potential energy, the second term is the free energy for $N$ independent harmonic oscillators, and $z_0(\beta,\rho)$ is the principal eigenvalue of some Ruelle operator that encodes the ``non-collision'' of certain Gaussian bridges; the last term is small at low density since $\lim_{\rho \to 0} z_0(\beta,\rho)=1$. Eq.~\eqref{eq:free-energy-z0} is the analogue of Eq.~(17) in~\cite{kunz1974one}.

If the system is non-neutral with fixed excess charge $Q = \rho L - N$, the bulk free energy $f(\beta,\rho)$ is given by~\eqref{eq:free-energy-z0} plus the term $- Q^2/(4 \rho)$. This term   represents the interaction between charges $Q/2$ placed at distance $b-a$: in Coulomb systems the excess charge typically accumulates at the boundary. This follows by combining our proofs with arguments given in \cite{kunz1974one}; these arguments also show that the bulk reduced density matrices in Theorem~\ref{thm:corr-funct}, which we now state, are unaffected by the excess charge.

\begin{theorem}\label{thm:corr-funct}
    \begin{enumerate}
    \item[(i)]  In the limit~\eqref{eq:thermolim} along $a,b\in \lambda \Z$, all reduced density matrices have uniquely defined limits
    \be
         \rho_n(x_1,\ldots,x_n;y_1,\ldots,y_n) = \lim \rho_n^{N}(x_1,\ldots,x_n;y_1,\ldots,y_n).
    \ee
    The convergence is uniform on compact subsets of $\R^n\times \R^n$, and the reduced density matrices $\rho_n^N$ and $\rho_n$ are continuous functions.
    \item[(ii)]
    The limit is periodic with respect to shifts by $\lambda=\rho^{-1}$,
    \be
        \rho_n(x_1 - \lambda ,\ldots;\ldots, y_n - \lambda ) = \rho_n(x_1,\ldots;\ldots,y_n)
    \ee
    for all $n\in \N$ and $\vect{x},\vect{y} \in \R^n$. Furthermore for every $\theta \notin \lambda\Z$ there is some $n\in \N$ and some $\vect{x} \in \R^n$ such that $\rho_n(\vect{x}-\theta;\vect{x}-\theta) \neq \rho_n(\vect{x};\vect{x})$.
    \end{enumerate}
\end{theorem}

The translation symmetry breaking of the reduced density matrices in part (ii) of the above theorem will follow from the symmetry breaking of probability measures on point configurations in the thermodynamic limit (Theorem \ref{thm:symbreak2} below) and adequately addressing a ``moment problem''. Let us make a couple further comments concerning part (ii), but before doing so, we
 rephrase the above theorem in terms of quantum states on $C^*$ algebras.

 Let $\mathcal{A}$ be the algebra of observables for fermions on the line, i.e., $\mathcal{A}$ is the CAR algebra (canonical anticommutation rules) over the one-particle Hilbert space $L^2(\R)$ \cite[Ch. 5.2]{bratteli1981operator}.
Let $\omega_N$ and $\omega$ be the states on $\mathcal{A}$ with reduced density matrices $\rho_n^N$ and $\rho_n$, respectively. Write $\tau_x:\mathcal{A}\to \mathcal{A}$ for the action of translation by $x$ on the observables.
The following is an immediate consequence of Theorem~\ref{thm:corr-funct} (locally uniform convergence of the reduced density matrices implies convergence of states on the CAR algebra, see \cite[Ch. 6.3.4]{bratteli1981operator}).

\begin{cor} \label{cor:states}
    In the thermodynamic limit along $a,b\in \lambda \Z$,  the states $(\omega_N)$ converge weakly  to $\omega$, i.e.,
    \be
        \forall A\in \mathcal{A}:\ \lim_{N\to \infty} \omega_N(A)= \omega(A).
    \ee
    The limit state is invariant under translations by integer multiples of $\lambda$,
    and $\lambda$ is the smallest period:
    \be
        \omega \circ \tau_x = \omega\ \Leftrightarrow \ x \in \lambda \Z.
    \ee
\end{cor}

Let us now make two remarks on further aspects of symmetry breaking.
Our first remark concerns the decay  of correlation functions and ergodicity. Part (c) of Lemma~\ref{lem:krein} below, together with arguments adapted from~\cite{kunz1974one}, shows that the $n$-particle densities $\rho_n(\vect{x};\vect{x})$ decay exponentially, e.g.,
\be \label{eq:two-point-decay}
    |\rho_2(x_1,x_2;x_1,x_2) - \rho_1(x_1;x_1) \rho_1(x_2;x_2) \bigr|  \leq C\exp(- \alpha |x_1-x_2| \bigr)
\ee
for suitable $\alpha, C>0$.
 We expect that in addition there is no off-diagonal long range order\footnote{A path integral picture is given in Figure 3 at the end of Section~\ref{sec:pi}.
 One must show that the relative probability of an ``open loop'' from $x$ to $y$ of winding number $n$, versus that of $n$ Brownian bridges with the same starting and ending points,
 goes to $0$.  This is intuitive since the probability that a Brownian motion is at a distance of order $n$ at time $n$ decays exponentially.}.  In other words we expect, for example, that $\rho_1(x;y) \to 0$ as $|y-x|\to \infty$; however, a proof would draw us too far from the objective at hand. But if this holds true, the limiting state $\omega$ is ergodic with respect to shifts by integer multiples of $\lambda$; in the absence of a proof, we know only that the restriction of $\omega$ to a commutative subalgebra of observables
-- described by the probability measure $\nu_\R^0$ defined below --  is ergodic (in fact, exponentially mixing).

The second remark concerns the appearance of reduced density matrices $\rho_n$ with $n\geq 2$ in part~(ii) of Theorem~\ref{thm:corr-funct}.
Brascamp and Lieb \cite{brascamp1975some} showed that at low densities, i.e. $\rho^{3/2} /\tanh(\beta \sqrt{\rho/2})$  sufficiently small\footnote{The Hamiltonian for a single charged particle is a harmonic oscillator. ``Small density"  means that the fluctuations of the harmonic oscillator / Ornstein-Uhlenbeck bridges are small compared to the typical interparticle spacing, $\sigma(\beta,\rho) \ll  \lambda = \rho^{-1}$ (see Appendix~\ref{app:ornstein}).},
the one-particle density $\rho_1(x;x)$ is a periodic function of $x$ with smallest period $\lambda$ (see also \cite[p. 314]{kunz1974one}).
Our result holds for all positive $\rho$ and $\beta$, but  leaves open the possibility that the period $\lambda$ is not visible at the level of the one-particle density -- in principle, a quantum state can have a non-trivial period but constant one-particle density, as illustrated by the next example.

\begin{example}
Let $\rho=\lambda=1$, $a=-b$, and let $\Psi_N$ be the many-fermion state given by the antisymmetrized product of $\chi_{-b},\ldots,\chi_{b-1}$ where $\chi_j(x)$ ($j\in \Z$) is the indicator that $x$ is in $[j,j+1)$. In the limit $b\to \infty$, the state is clearly not shift-invariant since the probability of seeing more than one particle in a small interval of width $\eps$ is zero if the interval is contained in a cell $[j,j+1)$, and non-zero if it intersects two distinct cells. Nevertheless, the one-particle density  $\rho_1(x;x) =\sum_{j=-\infty}^\infty \chi_j(x) =1$ is constant.
\end{example}

The quantum state on $\Hi_N$ determines a probability measure on point configurations as follows.
Denote the Weyl chamber by
\be \label{def:wc1}
W_N(a,b):= \{(x_1,\ldots,x_N) \mid a<x_1<\cdots <x_N<b\}.
\ee
Because of Eqs. ~\eqref{eq:corr-normalization} and ~\eqref{eq:corr-antisymmetry},
the $N$-particle density $\rho_N^N(\vect{x};\vect{x})$ integrates to $N!$ on $[a,b]^N$ and thus
\be\label{eq:density}
\int_{W_N(a,b)} \rho_N^N(\vect{x};\vect{x}) \dd\vect{x}=1.
\ee
We can thus think of
$\rho_N^N(\vect{x};\vect{x})$ as a probability measure on point configurations in $[a,b]$ (or equivalently, a point process).  To emphasize this perspective we rename this measure
 on configurations in $[a,b]$ as
\be\label{def:nab}
  \nab(\dd\vect{x})\equiv\rho_N^N(\vect{x};\vect{x}) \dd \vect{x},
\ee
we identify vectors  $\vect{x}$ with sets $\{x_1,\ldots,x_N\}$, and view $\nu_{a,b}^0$ as a measure on the space $\Omega$ consisting of locally finite subsets of $\R$ (every compact set contains at most finitely many points $x_j$).
The space $\Omega$ is equipped with the shift operator
$$\tau_u \vect{x} := \{ u+ x_i\mid x_i\in \vect{x}\}$$
and the topology (and Borel $\sigma$-algebra) generated by the continuous functionals $\vect{x}\mapsto \sum_{x_j\in \vect{x}} f(x_j)$ where $f$ runs over the continuous functions on $\R$ with compact support.
As always, we are interested in these measures in their thermodynamic limit~\eqref{eq:thermolim}.


\begin{theorem} [Symmetry breaking: point processes]  \label{thm:symbreak2}
    In the  limit~\eqref{eq:thermolim} along $a,b\in \lambda \Z$,  the measures $\nab$ converge weakly to a limiting probability measure $\nur$.
        The  measure $\nur$ is invariant under shifts $\tau_{n\lambda }$ of integer multiples of $\lambda $. Shifting  by a  non-integer multiple yields a measure which is singular with respect to $\nur$:
         $$u\notin \lambda  \Z \quad\text{implies}\quad \nur \circ\tau_{  u} \perp \nur.$$
\end{theorem}
Here ``$\perp$'' means as usual that the measures are mutually singular, and weak convergence means $\int f \dd \nu_{a,b}^0 \to \int f \dd \nu_\R^0$ for every continuous bounded function $f:\Omega \to \R$.

\section{Tools of the trade}\label{sec:tools}
\subsection{Path integrals}\label{sec:pi}

Recalling that $\rho={N}/{(b-a)}$, let
\be  \label{eq:xj}
    m_j := a + \lambda \(j-\frac{1}{2}\)\quad (j=1,\ldots,N).
\ee
It was noticed in \cite[Eq. 18]{baxter1963statistical} that for $a=0$ and  $x_1\leq \cdots \leq x_N$, the two sums in the potential energy of a configuration given in \eqref{eq:potential} can be written as
    \bea\label{eq:potential2}
        U(x_1,\ldots,x_N) &=&-\sum_{j=1} (2j-1-N) x_j + \rho \sum_{j=1}^N (x_j^2 -bx_j) +\text{const.}\nn\\
        &=&
        \rho \sum_{j=1}^N \bigl( x_j
        - m_j\bigr)^2 + \frac{N}{12 \rho}.
    \eea
An analogous identity holds for $a\neq 0$, and together with the Feynman-Kac formula,  allows us now to relate $\nab$ in \eqref{def:nab} to a Gaussian measure conditioned on a Weyl chamber.

Let $E$ be the space of continuous paths $\gamma:[0,\beta] \to \R$, equipped with the topology of uniform convergence and the corresponding Borel $\sigma$-algebra.
Let $\muxy$ be the measure on $E$ given by the non-normalized Brownian bridge measure, with total mass
$\muxy(E)= \P_\beta(x,y)$, where $\P_t(x,y)$ is the transition semi-group (heat kernel) for a standard Brownian motion in $\R$ with generator $\frac{1}{2}\Delta$.
Thus under $\muxy$, for all $0<t_1< \cdots<t_r <\beta$, the vector $(\gamma(t_1),\ldots,\gamma(t_n))$ has density
\be \nn
    \P_{t_1}(x,x_1) \P_{t_2-t_1}(x_1,x_2)\cdots \P_{\beta- t_r}(x_r,y).
\ee
and for $\muxy$-almost all $\gamma\in E$, $\gamma(0)=x$ and $\gamma(\beta) = y$.
Write $$\muxy(f) := \int_E f(\gamma) \dd \muxy(\gamma)$$
and  generalize the Weyl chamber \eqref{def:wc1} to a Weyl chamber for paths
\begin{align} \nn
     W_N^\beta(a,b)&:= \{ (\gamma_1,\ldots,\gamma_N)\in E^N \mid \forall t\in [0,\beta]\
             a<\gamma_1(t)<\cdots < \gamma_N(t)<b \},
\end{align}
see Figure~\ref{fig:bridges}.

\begin{lemma}[Feynman-Kac formula]  \label{lem:feynmankac}
    We have
    \bea\nn
        &&Z_N(\beta) = \\
        &&\int_{W_N(a,b)}  \mu_{x_1x_1} \times \cdots \times \mu_{x_N x_N}\Bigl(  e^{- \int_0^\beta U(\gamma_1(t),\ldots,\gamma_N(t))\dd t} \mathbf{1}_{W_N^\beta(a,b)}(\vect{\gamma}) \Bigr) \dd x_1\cdots \dd x_N \nn
    \eea
    and for all $\vect{x},\vect{y}\in W_N(a,b)$,
    \begin{multline}\nn
        \rho_N^N(\vect{x};\vect{y}) = \frac{1}{Z_N(\beta)}\,
            \mu_{x_1 y_1}\times \cdots \times \mu_{x_N y_N}
            \Bigl(
     e^{- \int_0^\beta U(\gamma_1(t),\ldots,\gamma_N(t)) \dd t} \mathbf{1}_{W_N^\beta(a,b)} (\vect{\gamma}) \Bigr).
    \end{multline}
\end{lemma}


The full proof is omitted as the lemma is a standard consequence of Fermi statistics and the usual Feynman-Kac formula \cite[Sec. 6]{simon1979functional}. However, let us briefly recall the general argument. First the antisymmetry is used to go from $\R^N$ to the Weyl chamber\footnote{For free fermions ($U \equiv 0$), Lemma~\ref{lem:feynmankac} reduces to the Karlin-McGregor determinantal formulas \cite{karlin1959coincidence} for non-coincidence probabilities of Brownian motions.}. The relevant boundary conditions are Dirichlet boundary conditions: indeed,  $\psi(x_2,x_1) = - \psi(x_1,x_2)$ yields $\psi(x_1,x_2) = 0$ whenever $x_1 = x_2$.
The Laplacian with Dirichlet boundary conditions  is the infinitesimal generator of a sub-Markov process, namely Brownian motion killed at the boundary of the Weyl chamber. The Feynman-Kac formula then gives a representation of the integral kernel of the Hamiltonian in the Weyl chamber and the lemma follows.

Note that in our notation, we have employed two equivalent ways to view the above result.  One is to think of $\vect{\gamma}(t)$ as a single $N$-dimensional Brownian bridge inside the Weyl chamber;
the other is to think of the components $\gamma_i(t)$ as non-colliding one-dimensional Brownian bridges as shown in the Figure 1.


\begin{figure}[h]
    \begin{center}
        \includegraphics[scale=.5] {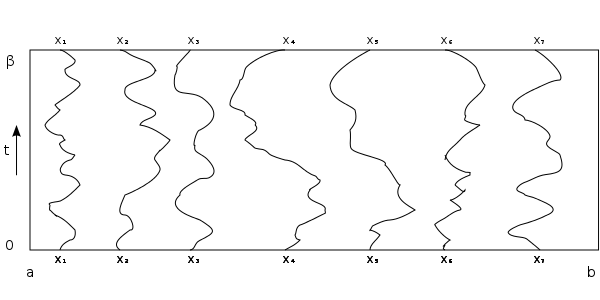}
        \begin{caption} {\footnotesize \label{fig:bridges}
            Non-colliding Brownian bridges:  a typical path configuration in $W_N^\beta(a,b)$, contributing to the partition function $Z_N(\beta)$.}
        \end{caption}
    \end{center}
\end{figure}

Using \eqref{eq:density} and \eqref{eq:potential2}, the Feynman-Kac formula gives us an alternative way of viewing the measure $\nab$:
        \bea \label{feynman}
    \nab(\dd \vect{x})=\frac{1}{Z_N(\beta)}\,
            \mu_{x_1 x_1}\times \cdots \times \mu_{x_N x_N}
            \Bigl(
     e^{- \int_0^\beta \sum_{j}(\gamma_j(t)-m_j)^2 \dd t} \mathbf{1}_{W_N^\beta(a,b)} (\vect{\gamma}) \Bigr)
    \eea
where by abuse of notation we write an equality between the measure $\nab$ and its density with respect to Lebesgue measure.
Our next lemma recasts \eqref{feynman} in a way that will prepare us to employ a Ruelle transfer operator. Let $\nu_{xy}$ be the measure on $E$ that is absolutely continuous with respect to $\muxy$ with Radon-Nikod{\'y}m derivative
\be \label{eq:rd_deriv}
    \frac{\dd \nuxy}{\dd \muxy}(\gamma) = \frac{1}{c(\beta,\rho)} \exp\Bigl( - \rho\int_0^\beta \gamma(t)^2 \dd t\Bigr)
\ee
with $c(\beta,\rho)$ chosen so that  the mixture
\be \label{eq:numeas}
     \nu := \int_\R \nu_{xx} \dd x
\ee
is a probability measure (in fact it is Gaussian since $\muxx$ is Gaussian and the exponent is quadratic; see Proposition \ref{prop:gaussian}  where we also compute $c(\beta,\rho)$ explicitly).
Define $\nu_j$ similarly but with $(\gamma(t)-m_j)^2$ replacing  $\gamma(t)^2$ in \eqref{eq:rd_deriv}.  Set
\be \label{def:nuhat}
    \hat\nu_{a,b}:= \nu_1\times \nu_2 \times \cdots\times \nu_N
\ee
which is by construction, a probability measure for $N$ \emph{independent} Gaussian paths.
\begin{lemma}\label{prop:nab}
    We have
    \be \label{eq:partfunction}
        Z_N(\beta) =c(\beta,\rho)^N e^{-\beta N /(12 \rho)}\,  \hat\nu_{a,b}\Bigl( W_N^\beta(a,b)\Bigr).
    \ee
    Moreover, if $ \nu_{a,b}(\dd\vect{\gamma})$ is defined as the measure $ \hat\nu_{a,b}(\dd\vect{\gamma})$ conditioned on the event $\vect{\gamma} \in W_N^\beta(a,b)$, i.e.,
    \be \nn
        \frac{\dd  \nu_{a,b}}{\dd \hat\nu_{a,b}} (\gamma_1,\ldots,\gamma_N)
            := \frac{1}{\hat\nu_{a,b}(W_N^\beta(a,b))}\, \mathbf{1}_{W_N^\beta(a,b)} (\gamma_1,\ldots,\gamma_N),
    \ee
    then the law of $(\gamma_1(0),\ldots,\gamma_N(0))$ under $\nu_{a,b}(\dd\vect{\gamma})$ has a density with respect to Lebesgue measure on $\R^N$ given by $\nab(\dd\vect{x})$ as defined in \eqref{def:nab}.
\end{lemma}
\begin{proof}
    Eq. \eqref{eq:partfunction} follows from \eqref{eq:potential2}, Lemma \ref{lem:feynmankac}, and the definition of $\hat\nu_{a,b}$ given by \eqref{eq:rd_deriv}-\eqref{def:nuhat}.
    The statement concerning $\nab$ follows from \eqref{feynman}.
\end{proof}


The Feynman-Kac formulation also gives us a nice representation for the reduced density matrices, which is a variant of some well-known functional integral representations (see \cite{ginibre1965reduced} or \cite[Ch. 6.3.3]{bratteli1981operator}). We give an expression and proof only for the one-particle matrix and content ourselves with a geometric description for the $n$-particle matrices.

Fix $x,y\in [a,b]$ with $x \leq y$ and fix $j,k\in\{1,\ldots,N\}$ such that $j\le k$. Let $\Gamma_{xjky}\subset W_N^\beta(a,b)$ (see Figure~\ref{fig:bridges2}) be the set of non-intersecting paths $(\gamma_1,\ldots,\gamma_N)$ such that
\begin{itemize}
    \item $\gamma_i(0) = \gamma_i(\beta)$ for all $i < j$ and $i>k$,
\item $\gamma_i(\beta) = \gamma_{i+1}(0)$ for $i=j,\ldots,k-1$  (if $j=k$ this condition is vacuous),
    \item $\gamma_j(0) = x$ and $\gamma_k(\beta) = y$.
\end{itemize}
Also, we label the starting point of the $i$th path as $x_i$ so that in particular $x=x_j$.

\begin{figure}[h]
    \begin{center}
        \includegraphics[scale=.5] {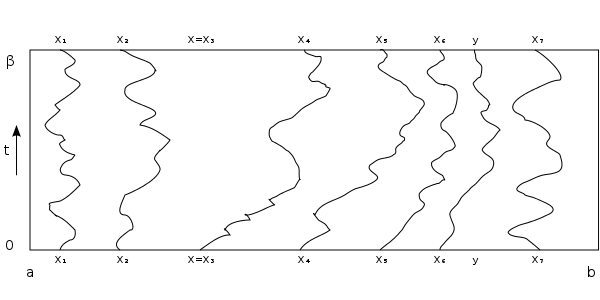}
        \begin{caption} {\footnotesize \label{fig:bridges2}
            A path configuration in the set $\Gamma_{xjky}$ with $j=3$ and $k= 6$ contributing to the one-particle reduced density matrix $\rho_1^N(x;y)$.
            }
        \end{caption}
    \end{center}
\end{figure}

For $x=x_j$ and $\vect{x} = (x_1,\ldots,x_N)\in W_N(a,b)$, let
\bea
\nn &&W_{xjky}:=\{(\vect{x},y): a<x_1<\cdots<x_{k}<y<x_{k+1}<\cdots<x_{N}<b,\ x_j=x\}.
\eea
By integrating over $W_{xjky}$, on $\Gamma_{xjky}$ we define the measure
\bea  \label{eq:loop-meas}
    &&\nu_{xjky}(\dd\vect{\gamma}) :=\\
    &&\nn\int_{W_{xjky}}  \mathbf{1}_{W_N^\beta(a,b)}\nu_{x_1 x_1}(\dd \gamma_1)\times \cdots \Bigl( \nu_{xx_{j+1}}\times \cdots \times \nu_{x_ky} \Bigr) \cdots \times \nu_{x_{N} x_{N}}(\dd\gamma_N) \,\dd x_1 \cdots \widehat{\dd x_{j}} \cdots \dd x_{N}
\eea
where  $\widehat{\dd x_{j}}$ signifies that integration over this variable is omitted.
The above is a mixture of bridge measures obtained by integrating out the free starting and ending points.

Because of the self-adjointness of $\exp(- \beta H_N)$, the reduced density matrices are symmetric with respect to the starting and ending points $(\vect{x},\vect{y})$ (this is different from \eqref{eq:corr-antisymmetry}); in particular,  $\rho_1^N(x;y) = \rho_1^N(y;x)$ and we need only treat the case $x\leq y$.

\begin{lemma}\label{lem:reduced-matrix}
    Let $x,y\in [a,b]$ with $x\leq y$.  The reduced one-particle matrix is given by
    \be  \label{eq:reduced-matrix}
        \rho_1^N(x;y) = \frac{1}{\hat \nu_{a,b}(W_N^\beta(a,b))}
            \sum_{j\le k} (-1)^{k-j} \nu_{xjky} \bigl( W_N^\beta(a,b) \bigr).
    \ee
\end{lemma}

\begin{remark}
 When $x=y$, we must have $j=k$. Therefore  the one-particle density $\rho_1(x;x)$ has a much simpler expression as a sum of one-dimensional marginals of bridge starting  points $\gamma_j(0)$, $j=1,\ldots,N$, under $\hat \nu_{a,b}$. In particular, Eq.~\eqref{eq:reduced-matrix} is compatible with $\int_a^b \rho_1^N(x;x) \dd x = N$.
\end{remark}

\begin{proof}
    Consider first the case $N=2$.
     Fix $x,y\in \R$ and assume $x\leq y$.
     The one-particle matrix is by definition proportional to
    \be
        \rho_1^2(x;y) \propto   \int_a^b \rho_2^2 (x, x'; y,x') \dd x'.
    \ee
    We use the antisymmetry~\eqref{eq:corr-antisymmetry} to reorder the arguments of $\rho_2^2$ in the integrand and obtain
    \be \label{eq:sum-3}
         \int_a^x \rho_2^2(x',x;x',y) \dd x' -  \int_x^y \rho_2^2(x,x';x',y) \dd x'
            + \int_y^b \rho_2^2(x,x';y,x') \dd x.
    \ee
    We may now apply  Lemma~\ref{lem:feynmankac}: the first term corresponds to paths $(\gamma_1,\gamma_2)$ with $\gamma_1(0) = \gamma_1(\beta) = x'$, $\gamma_2(0) = x$, and $\gamma_2(\beta) = y$. The third term is similar, except for a switch in the roles of $\gamma_1$ and $\gamma_2$. The middle term corresponds to paths with $\gamma_1(0) = x$, $\gamma_1(\beta) = x' = \gamma_2(0)$, and $\gamma_2(\beta) = y$. Thus we find
    \be
        \rho_1^2(x;y) = \frac{1}{C}  \Bigl( \nu_{x11y}  - \nu_{x12y} + \nu_{x22y} \Bigr)\bigl( W_N(a,b)\bigr)
    \ee
    for some constant $C>0$. The proof of Eq.~\eqref{eq:reduced-matrix} is concluded by
    computing $C>0$ via the condition $\int_a^b \rho_1^2(x;x) \dd x = 2$.

    The computation for general $N\in \N$ is similar; the sign $(-1)^{k-j}$ comes from reordering the arguments in the sector where there are $k-j$ variables $x'_i$ between $x$ and $y$.
\end{proof}

Lemma~\ref{lem:reduced-matrix} has analogues for the $n$-particle matrices, which we now briefly describe. In Figure 3 below, we show a common geometric picture, c.f. \cite[Ch. 6.3.3]{bratteli1981operator}: we view $[a,b]\times [0,\beta]$ as a cylinder, with $\beta$ the periodic (angular) coordinate, and think of paths $\gamma\in E$ as loops of winding number $n=1$. If  $\gamma(0)= \gamma(\beta)$, the loop is \emph{closed}, otherwise it is \emph{open}. More generally, a loop of winding number $w\geq 2$
is a vector $(\gamma_1,\ldots,\gamma_w)$ such that $\gamma_j(\beta) = \gamma_{j+1}(0)$ for all $j=1,\ldots,w-1$. It is closed if $\gamma_w(\beta) = \gamma_1(0)$.
Note that we may also write the loop as a single ``composite'' path $\gamma:[0,n\beta] \to [a,b]$, where
$\gamma(j\beta +t) =\gamma_j(t)$ for every $j=1,\ldots,N-1$ and $t\in [0,\beta]$. Loops are not allowed to have self-intersections (i.e., $\gamma_1,\ldots,\gamma_w$ do not collide), but they can wind from right to left: we may have $\gamma_w(\beta) < \gamma_1(0)$.

\begin{figure}[h]\label{fig:bridges3}
    \begin{center}
        \includegraphics[scale=.5] {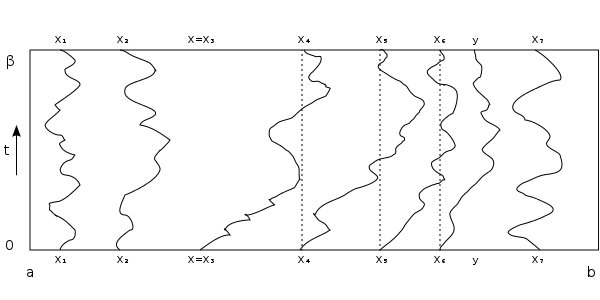}
        \begin{caption} {\footnotesize
        An ``open loop'' from $x$ and $y$ with winding number 4. The endpoints of dotted lines are identified.}
        \end{caption}
    \end{center}
\end{figure}

Fix $x_1<\cdots<x_n$, $y_1,\ldots,y_n$ (the $y_k$'s need not be ordered).  The paths $(\gamma_1,\ldots,\gamma_N)$ contributing to the representation of the $n$-particle matrix consist of
\begin{itemize}
    \item  $n$ open loops of respective winding numbers $w_1,\ldots,w_n\geq 1$. The $k$-th open loop starts at $x_k$ and ends at $y_k$.
    \item $N- \sum_k w_k$ closed loops of winding number $1$.
    \item The loops do not collide and have no self-intersections.
\end{itemize}
Open loops can be entwined -- we may have, for example, an open loop of winding number $2$ going from $x_1$ to $x'$ and then $y_1$, and another open loop going from $x_2\in (x_1,x')$ to $y_2\in (x', y_1)$. In this example, $x_1<x_2$ but $y_2<y_1$.

The representation of the $n$-particle reduced density matrix is in terms of the natural analogs of Eqs.~\eqref{eq:loop-meas} and~\eqref{eq:reduced-matrix}. Each open loop configuration comes with the sign
$
    \prod_{i=1}^n (-1)^{w_i - 1}.
$

\subsection{ Krein-Rutman theorem}

As mentioned in the introduction, following \cite{kunz1974one} and \cite{aizenman1980structure}, the Krein-Rutman theorem \cite{krein1948linear} is one of our main tools for proving symmetry breaking.  It is a generalization of the Perron-Frobenius theorem.

Before stating a version of the Krein-Rutman theorem, let us first describe the
 Ruelle  operator to which it will be applied.   The Ruelle operator will be used to reproduce the probability measure \eqref{feynman} on point configurations $\vect{x}$ in $[a,b]$. By Lemma \ref{prop:nab}, it is in fact enough to produce the measure $\nu_{a,b}(\dd \vect{\gamma})$, which we now do.
Our kernel $K:E\times E \to \R$ operates on $L^2(E)$. 
It is defined  by
\be
    K(\gamma,\eta) := \begin{cases}
                            1,& \quad \forall t\in [0,\beta]:\, \gamma(t) < \eta(t)+\lambda ,\\
                            0,&\quad \text{else}.
                            \end{cases}
\ee
We let $\K$ be the associated integral operator on the separable Hilbert space $L^2(E,\nu)$ with $\nu$ defined as in \eqref{eq:numeas},
\be \label{eq:ruelle operator}
    (\K f)(\gamma) := \int_{E} K(\gamma,\eta) f(\eta) \nu(\dd \eta).
\ee
Its adjoint is
\be \label{eq:adjoint}
    (\K^* f)(\gamma) := \int_{E} K(\eta,\gamma) f(\eta) \nu(\dd \eta).
\ee
Write  $\la f,g\ra:= \int_{E} \overline{f} g \,\dd \nu$ for the scalar product in $L^2(E,\nu)$ and $||f||:= \sqrt{\la f,f\ra}$ for the $L^2$-norm.

\begin{lemma}\label{lem:FG}
    Let
    \bea\nn
        F(\gamma)&:=&K(- 1/(2\rho),\gamma)=\mathbf{1}_{\{\forall t\in [0,\beta]:\, \gamma(t) >- 1/(2\rho)\}}\\
        \nn G(\gamma)&:=&K(\,\gamma\,, 1/(2\rho)=\mathbf{1}_{\{\forall t\in [0,\beta]:\, \gamma(t) < 1/(2\rho) \}}.
    \eea
    We have
    \be\nn
        \hat\nu_{a,b}\bigl( W_N^\beta(a,b) \bigr) = \la F, \K^{N-1} G\ra .
    \ee
\end{lemma}

\begin{proof}
Recall the definition of $m_j$ from \eqref{eq:xj}.  Note that $$a<\gamma_1(t)<\cdots <\gamma_N(t)<b$$ if and only if the shifted paths $\tilde \gamma_j(t):= \gamma_j(t) - m_j$ satisfy

\be\tilde \gamma_1(t)> a - m_1= - 1/(2\rho)\ \text{ and }\ \tilde \gamma_N(t) < b-m_N = 1/(2\rho),\ee
and for all $j \in \{1,\ldots,N-1\}$,
    \be \label{eq:tilde gamma}\nn
    \tilde \gamma_j(t) < \tilde \gamma_{j+1}(t) + m_{j+1} - m_j = \tilde \gamma_{j+1}(t) + \lambda .
    \ee
    From Lemma \ref{prop:nab}  and \eqref{eq:numeas} we have that
    \bea\label{eq:FG}
        &&\hat\nu_{a,b}\Bigl( W_N^\beta(a,b) \Bigr) \\
        \nn&=& \int_{a<\gamma_1<\cdots<\gamma_N<b}\nu_1(\dd\gamma_1) \cdots \nu_N(\dd\gamma_N) \\
            \nn &=&\int_{E^N}  K(- 1/(2\rho),\tilde\gamma_1) K(\tilde \gamma_1,\tilde \gamma_2) \cdots K(\tilde \gamma_{N-1},\tilde \gamma_N) K(\tilde \gamma_N,1/(2\rho)) \nu(\dd \tilde \gamma_1) \cdots \nu(\dd \tilde \gamma_N) \\
        \nn&=&\int_{E^N}  F(\tilde \gamma_1) K(\tilde \gamma_1,\tilde \gamma_2) \cdots K(\tilde \gamma_{N-1},\tilde \gamma_N) G(\tilde \gamma_N) \nu(\dd \tilde \gamma_1) \cdots \nu(\dd \tilde \gamma_N) \\
        \nn&=& \la F, \K^{N-1} G \ra.
    \eea
\end{proof}


We are now ready for our version of the Krein-Rutman theorem~\cite{krein1948linear}. It follows from the standard theorem by simple arguments adapted from~\cite[Appendix A]{kunz1974one}. In particular, the Cauchy-Schwarz inequality allows us to transfer properties from the Hilbert space $L^2(E,\nu)$ to the space of bounded functions. Our statements are uniform in $\gamma$ and do not involve $\nu$-null sets. This is important because the representation of reduced density matrices uses paths from $x$ to $y\neq x$, which form a $\nu$-null set.


\begin{lemma}[Krein-Rutman]\label{lem:krein}
    Let $z_0:=||\K||>0$. Then:
    \begin{enumerate}
    \item[(a)]  There is a unique strictly positive function $\Psi_0: E\to \R$ such that  $(\K \Psi_0)(\gamma) = z_0 \Psi_0(\gamma)$ for all $\gamma \in E$ and
     $\int_E \Psi_0(\gamma)  \Psi_0(- \gamma) \nu(\dd \gamma) =1$.
    \item[(b)] The reflected function $\tilde \Psi_0 (\gamma) := \Psi_0(-\gamma)$ satisfies $\K^* \tilde \Psi_0 =z_0 \tilde \Psi_0$.
    \item[(c)] $\Psi_0$ is bounded.
    \item[(d)]  For suitable $\eps,C>0$, all $f\in L^2(E,\nu)$ and all $n\in \N$,
    \be
    \begin{aligned}
         \Bigl| \frac{1}{z_0^n} \K^n f(\gamma) - \la\tilde \Psi_0, f\ra  \Psi_0 (\gamma)\Bigr| &\leq C e^{-\eps n} ||f||,  \\
           \Bigl| \frac{1}{z_0^n}( \K^*)^n f (\gamma) - \la\Psi_0, f\ra  \tilde \Psi_0(\gamma) \Bigr| &\leq C e^{-\eps n} ||f||.
    \end{aligned}
    \ee
    \end{enumerate}
\end{lemma}

\begin{proof}
 (a) One can easily check that our operator is Hilbert-Schmidt and irreducible, and that it maps non-negative functions to non-negative functions.  The Krein-Rutman theorem \cite{krein1948linear} shows that $z_0$ is a simple eigenvalue, and the eigenfunction can be chosen strictly positive. Hence there is a unique  $\Psi_0 \in L^2(E,\nu)$ such that $\int_E \Psi_0(-\gamma) \Psi_0(\gamma) \nu(\dd \gamma) =1$,  $\Psi_0(\gamma) >0$ and $(\K\Psi_0)(\gamma) = z_0 \Psi_0(\gamma)$ for $\nu$-almost all $\gamma$. Asking that the last equality holds for all $\gamma\in E$ removes the ambiguity on $\nu$-null sets.
\\
    Part (b) follows from the  symmetry $K(\gamma,\tilde \gamma) = K(- \tilde \gamma, -\gamma)$.    For later purposes we also note the following: the projection $|\Psi_0\ra \la \tilde \Psi_0|:\ f\mapsto \la \tilde \Psi_0, f\ra \Psi_0$ satisfies $\K |\Psi_0\ra \la \tilde \Psi_0| = |\Psi_0\ra \la \tilde \Psi_0| \K = z_0 |\Psi_0\ra \la \tilde \Psi_0|$.
    \\
    (c) The Cauchy-Schwarz inequality yields
    \be
        \Psi_0(\gamma) = z_0^{-1} \int_E K(\gamma,\tilde \gamma) \Psi_0(\tilde \gamma) \nu(\dd \tilde \gamma) \leq z_0^{-1} ||\Psi_0|| (\int_E K(\gamma,\tilde \gamma)^2 \nu(\dd \tilde \gamma) )^{1/2}.
    \ee
\\
    (d)
    The spectrum of $\K$  consists of eigenvalues only because $\K$ is compact. We have also just shown that $z_0$ is a simple eigenvalue of $\K$ and that every other eigenvalue has a strictly smaller absolute value. Let
    \be
        z_1:=\max \Bigl \{ |\lambda| : \lambda \in \sigma(\K),\ \lambda \neq z_0 \}<z_0 \Bigr \}
    \ee
    The operator $\K - z_0 |\Psi_0\ra \la \tilde \Psi_0|$ is compact and has spectral radius $z_1$,
    \be
        \lim_{n\to \infty} \frac{1}{n} \log || \bigl(\K - z_0 |\Psi_0\ra \la \tilde \Psi_0|\bigr)^n || =  z_1.
    \ee
    Since $$\K (\K - z_0 |\Psi_0\ra \la \tilde \Psi_0|) = (\K - z_0 |\Psi_0\ra \la \tilde \Psi_0|) \K = 0,$$ we have $$\K^n = (\K - z_0 |\Psi_0\ra \la \tilde \Psi_0|)^n + z_0^n |\Psi_0\ra \la \tilde \Psi_0|,$$ and we deduce that for suitable $\eps,C'>0$ and all $n\in \N$,
    \be \label{eq:kmix}
         || \frac{1}{z_0^n} \K^n - |\Psi_0\ra \la \tilde \Psi_0|\,  ||\leq C' \exp( - \eps n).
    \ee
    The proof is completed by applying the Cauchy-Schwarz inequality as in the proof of (c) and then applying the inequality~\eqref{eq:kmix}
    \be
    \begin{aligned}
        \Bigl| \bigl(\frac{1}{z_0^n} \K^n f - |\Psi_0\ra \la \tilde \Psi_0| f\bigr)(\gamma)  \Bigr|
            & = \Bigl| \frac{1}{z_0} \K \bigl( \frac{1}{z_0^{n-1}} \K^{n-1} - |\Psi_0\ra \la \tilde \Psi_0|\bigr) f(\gamma) \Bigr|\\
        &\leq  ||\bigl( \frac{1}{z_0^{n-1}} \K^{n-1} - |\Psi_0\ra \la \tilde \Psi_0|\bigr) f(\gamma) || \\
        & \leq  C' e^{- \eps(n-1)} ||f||.
    \end{aligned}
    \ee
    The proof for the adjoint is similar.
\end{proof}

\section{Symmetry breaking}\label{sec:paul}

\begin{proof}[Proof of Theorem \ref{thm:symbreak2}]
We start with the existence of the limiting probability measure for \emph{labeled} particles; the existence of $\nu_\R^0$ (on particles without labels, i.e., a point process) will be a direct consequence.
Using Lemma \ref{prop:nab} and Kolmogorov's extension theorem, it is enough to prove  the existence along \eqref{eq:thermolim} of $\nu_\R:=\lim\nu_{a,b}$ in the sense of weak convergence of the finite-dimensional cylinder distributions.  For this to make sense, we must first relabeling the bridges so as to view $\nu_{a,b}$ as a measure on $E^{\{a\rho +1,\ldots,b \rho\}}$ rather than $E^{\{1,\ldots,N\}}$.

More specifically, we relabel each of the $N$ conditioned Brownian bridge measures using sub-indices $\ell_j,j\in\Z$.
The sub-index shifts the original index by $a\rho$ so that
\be
    \ell_{a\rho+1}=1,\ldots, \ell_0=-a\rho, \ldots, \ell_{a\rho+N}=\ell_{b\rho}=N
\ee
and in general $\ell_{j}=-a\rho+j$.   Note that in the limit~\eqref{eq:thermolim}, we have $a\rho +1\to -\infty$ and $ a\rho + N = b\rho  \to \infty$, and the lattice points defined in~ \eqref{eq:xj} satisfy
\be  \label{eq:mlj}
        m_{\ell_j} = (j- 1/2) \lambda.
\ee
We will show that for every fixed $(j_1,\ldots,j_n) \in \Z^n$, the law of $(\gamma_{\ell_{j_1}},\ldots,\gamma_{\ell_{j_n}})$ under $\nu_{a,b}$ converges weakly to a measure on $E^{\{j_1,\ldots,j_n\}}$. The family of measures obtained in this way satisfies the consistency conditions required by Kolmogorov's extension theorem and hence are the cylinder marginals of a uniquely defined measure $\nu_\R$ on $E^\Z$. In this sense, we obtain $\lim \nu_{a,b}=  \nu_\R$.

 For simplicity, we show the convergence of cylinder marginals  only for the single
marginal $\gamma_{\ell_1}$ since the general argument is similar (note that $m_{\ell_1}$ is the smallest positive $m_j$ as defined in \eqref{eq:xj}, see ~\eqref{eq:mlj}).  We will leave it to the reader to confirm that the limits are consistent in the sense required by Kolmogorov's theorem.

We show pointwise convergence to a probability density, under the limit \eqref{eq:thermolim}, which implies weak convergence of the distributions of $\gamma_{\ell_1}$.  Fix $\bar\gamma\in E$ and  for a system of $N$ bridges
recall that $\tilde\gamma_\ell(t)= \gamma_\ell(t)-m_\ell$.
By Lemma \ref{prop:nab}, we have a density, with respect to $\nu_{\ell_1}$, for the $\ell_1$-th bridge given by
\bea
&&f_{a,b}(\bar\gamma):=\nn\\
 &&\frac{1}{\hat\nu_{a,b}(W_N^\beta(a,b))}\int_{a<\gamma_1<\cdots<\gamma_{\ell_{0}}<\bar\gamma<\gamma_{\ell_2}<\cdots<\gamma_N<b} \dd\nu_1 \cdots \dd\nu_{\ell_{0}}\dd\nu_{\ell_2}\cdots \dd\nu_N. \nn
 \eea
Using the adjoint  $\K^*$ defined in \eqref{eq:adjoint} and  the  arguments in  Lemma \ref{lem:FG}, we obtain the shifted density with respect to $\nu$:
\bea \tilde f_{a,b}(\bar\gamma)&=&\frac{1}{\hat\nu_{a,b}(W_N^\beta(a,b))} ((\K^*)^{\ell_1-1} F)({\bar\gamma}) (\K^{N-\ell_1} G)(\bar\gamma)\nn\\
&=&\frac{((\K^*)^{\ell_1-1} F)(\bar\gamma) (\K^{N-\ell_1} G)(\bar\gamma)}{\langle F,\K^{N-1} G\rangle}.\nn
\eea
By Lemma \ref{lem:krein}, since $\ell_1\to\infty$ under the limit \eqref{eq:thermolim}, we have \be \label{eq:one-marginal-lim}
\tilde f_{a,b}(\bar\gamma)
\longrightarrow \frac{\langle F, \Psi_0\rangle\tilde\Psi_0(\bar\gamma)\langle \tilde\Psi_0, G\rangle\Psi_0(\bar\gamma)}{\langle F,\Psi_0\rangle \langle \tilde\Psi_0, G\rangle} =\tilde\Psi_0(\bar\gamma)\Psi_0(\bar\gamma)=:\tilde f(\bar\gamma)
\ee
where $\Psi_0$ is the positive eigenvector associated to the largest eigenvalue $z_0$ of $\K$, $\tilde \Psi_0(\gamma):= \Psi_0(-\gamma)$, and $\la \tilde \Psi_0,\Psi_0\ra =1$ as in Lemma~\ref{lem:krein}. 
Note that $\tilde f(\gamma)=\tilde f(-\gamma)$ is even. 
 Note also that the value of $\ell_1$ appears in the  argument only through the fact that
$\ell_1\to\infty$ under \eqref{eq:thermolim}. Thus, the single-bridge marginal distributions are shifts of each other under integer multiples of  $\lambda $.
  As mentioned earlier, the above argument can be extended to show the existence of other limiting cylinder marginals of $\nu_\R$ as well as the fact that they are translations of each other under the shifts $\lambda \Z$.

 Thus we have shown that the measures for
bridges labeled by the shifted indices converge to a measure
$\nu_\R$ on $E^\Z$ for infinitely many bridges. By
Lemma~\ref{prop:nab}, it follows that the measure for labeled
particles converges to the law of $\bigl( \gamma_j(0) \bigr)_{j\in
\Z}$ under $\nu_\R$. From this one can deduce that $\nu_{a,b}^0$,
which is the law of $\{\gamma_1(0),\ldots,\gamma_N(0)\}$ under
$\nu_{a,b}$, converges weakly to the law of $\{\gamma_j(0) \mid j
\in \Z\}$ under $\nu_\R$. A technicality arises because the
map $\R^\Z \to \Omega$, $(x_j)_{j\in \Z} \mapsto \{x_j\mid j \in
\Z\}$ is not, in general, continuous. For example, it may not map a finite number of points to a finite interval. But one can check using \eqref{eq:nu-bound} that in our case, 
the mapping is a.s. continuous.  We leave the details to the reader.

We turn now to the mutual singularity of measures shifted by non-integer multiples of $\lambda$. 
It follows from Eq.~\eqref{eq:one-marginal-lim}, the evenness of
$\tilde f(\tilde \gamma) = \tilde f(- \tilde \gamma)$, and the
evenness of the reference Gaussian measures from
Eq.~\eqref{eq:numeas}, that the law of $\gamma_j(0) $ is invariant
under reflections around $(j-1/2) \lambda$. Thus the random variable
$Y_j:( E^\Z,\nu_\R) \to \R$ defined by
\be \nn
    Y_j(\vect{\gamma}):=  \gamma_j(0) - (j-1/2) \lambda
\ee
satisfies
\be \label{eq:means}
    \E Y_j = 0.
\ee
Let $\tau_\lambda(\vect{\gamma}) := \bigl( \gamma_{j+1} - \lambda \bigr)_{j\in \Z}$, so that
\be
   \nn Y_j(\vect{\gamma}) = Y_{j-1}(\tau_{\lambda }(\vect{\gamma})).
\ee
The sequence of random variables $(Y_j)_{j\in\Z}$ is stationary, and by Lemma \ref{lem:krein} and standard arguments, is also ergodic (in fact mixing). Thus, we have that $\nu_\R$-a.s., for every $k\in\Z$,
\be \label{eq:yj}
    \lim_{n\to\infty}\frac 1 n \sum_{j=1}^n Y_{k+j}=0.
\ee
Intuitively, we also have that $\nu_\R \circ\tau_{\lambda  u}$-a.s.,
\be \label{eq:y'j}
\lim_{n\to\infty}\frac 1 n \sum_{j=1}^n Y_{k+j} \equiv u \ \ (\text{mod}\, \lambda )
\ee
which proves the mutual singularity.  The rest of the proof is devoted to rigorously proving \eqref{eq:y'j}.

 For $\vect{\gamma}\in E^\Z$, let $k(\vect{\gamma}):= \min\{ j
\in \Z\mid \gamma_j(0) \geq 0\}$. One can  check that such a
$k(\vect{\gamma})$ exists a.s. using, for example, the
Borel-Cantelli lemma. Set \be \nn    Y'_j(\vect{\gamma}):=
\gamma_{k(\vect{\gamma}) +j}(0) - (j-1/2)\lambda. \ee We have \be \nn
    \frac{1}{n} \sum_{j=1}^n Y'_j (\vect{\gamma})
        = \frac{1}{n} \sum_{j=0}^{n-1} Y_{k(\vect{\gamma})+j } - k(\vect{\gamma})
 \lambda
\ee
thus
\be
    \lim_{n\to \infty} \exp\Bigl( \mathrm{i} \frac{2\pi}{\lambda} \frac{1}{n} \sum_{j=1}^n Y'_j \Bigr)  = 1\quad \nu_\R\text{-a.s.}
\ee
Since $Y'_j$ is a function of the set of starting points, we may rewrite the latter identity in terms of the point process: label the points $x_j = x_j(\omega)$ in a configuration $\omega\subset \R$ as  $\omega= \{ \cdots < x_0<0 \leq x_1<x_2<\cdots\}$. We have
\be \label{eq:plim}
    \lim_{n\to \infty} \exp\Bigl( \mathrm{i} \frac{2\pi}{\lambda} \frac{1}{n} \sum_{j=1}^n \bigl( x_j(\omega)- (j-1/2) \lambda\bigr) \Bigr)  = 1\quad \nu_\R^0\text{-a.s.}
\ee
 Let ${u} \in \R$ and $m=m({u},\omega)$ be such that $x_m (\omega)\leq - {u} <x_{m+1}(\omega)$. We have
\begin{multline}
    \exp\Bigl( \mathrm{i} \frac{2\pi}{\lambda} \frac{1}{n} \sum_{j=1}^n \bigl( x_j(\tau_{u} \omega)- (j-1/2) \lambda\bigr) \Bigr) \\
    = \exp\Bigl(\mathrm{i}  \frac{2\pi{u} }{\lambda} \Bigr) \times \exp \Bigl( \mathrm{i} \frac{2\pi}{\lambda} \frac{1}{n} \sum_{j=1}^n \bigl( x_{j+m({u},\omega)} (\omega)- (j-1/2) \lambda\bigr)\Bigr).
\end{multline}
An argument similar to the proof of Eq.~\eqref{eq:plim} shows that the second factor on the right side converges to $1$, $\nu_\R^0$-almost surely. It follows that
\be
    \lim_{n\to \infty} \exp\Bigl( \mathrm{i} \frac{2\pi}{\lambda} \frac{1}{n} \sum_{j=1}^n \bigl( x_j(\omega)- (j-1/2) \lambda\bigr) \Bigr)  = \exp\Bigl(\mathrm{i}  \frac{2\pi{u} }{\lambda} \Bigr) \quad (\nu_\R^0\circ \tau_{u})\text{-a s.}
\ee

%
\end{proof}

\paragraph{Remarks:}
\begin{enumerate}
\item[1.]
An elementary proof of a weaker version of Theorem \ref{thm:symbreak2} follows by applying Theorem 1.9 of \cite{brascamp1975some} to the expression \eqref{feynman} to conclude
tightness for a marginal distribution of a single conditioned Brownian bridge.  Then one can appeal to Theorem 2.1 of \cite{aizenman2001bounded}.  See also \cite[Sec. 4.2]{aizenman2010symmetry}.
\item[2.] A proof using the electric field as in \cite{aizenman1980structure, aizenman2010symmetry} is also possible.  Such a route would however require the vanishing of  volume averages for the electric field (see \cite[Thm 3.4]{aizenman2010symmetry}). The easiest way to achieve this, that the authors are aware of, is via the ergodic theorem which brings us back to the Ruelle operator and Krein-Rutman theorem (see \cite[Proof of Lemma 4]{aizenman1980structure}).
\end{enumerate}

\section{Free energy and reduced density matrices} \label{sec:free energy}

First we prove the result on the asymptotics of the partition function.

\begin{proof}[Proof of Theorem~\ref{thm:free-energy}] \label{proof:free-energy}
By \eqref{eq:tilde gamma}, we have
    \be
    \begin{aligned}
          \hat\nu_{a,b}\bigl(W_N^\beta(a,b)\bigr)&= \la F, \K^{N-1} G \ra \\
         & = z_0^{N-1} \la F, \tilde \Psi_0 \ra \la \Psi_0, G\ra + \la F, (\K - z_0 |\Psi_0\ra \la \tilde \Psi_0|)^{N-1} G\ra \\
        & = z_0^{N-1} \la F, \tilde \Psi_0 \ra \la \Psi_0, G\ra \Bigl( 1+ O \bigl( e^{-\eps N}  \bigr) \Bigr)
    \end{aligned}
    \ee
    with $\eps>0$ as in  Lemma~\ref{lem:krein}.
    Combined with  Eqs.~\eqref{eq:partfunction} and ~\eqref{eq:hotrace}, this yields
    \begin{multline}
        \log Z_N(\beta) = N\left( -  \frac{\beta }{12\rho} -  \log \bigl( 2\sinh( \beta \sqrt{{\rho}/{2}} \bigr) + \log z_0 \right) \\
            -  \log z_0 + \log \la F, \tilde \Psi_0\ra+ \log \la \Psi_0, G\ra + O \bigl( e^{- \eps N} \bigr).
    \end{multline}
\end{proof}

Next we come to the existence of  reduced density matrices in the thermodynamic limit and to symmetry breaking.

\begin{proof} [Proof of Theorem~\ref{thm:corr-funct}]
    (i) We start with the one-particle matrix. Let $x\leq y$, $j\le k$,  and $\nu_{xjky}$ as in  Lemma~\ref{lem:reduced-matrix}. The quotient $\nu_{xjky}\bigl( W_N^\beta(a,b)\bigr)/\hat \nu_{a,b}\bigl( W_N^\beta(a,b)\bigr)$ is equal to the integral of
    \be \label{eq:qjk}
        \frac{\bigl( (\K^*)^{j-1} F\bigr)(\gamma_j) \times
             K(\gamma_j,\gamma_{j+1}) \cdots K(\gamma_{k-1},\gamma_k)
        \times \bigl( \K^{N-k}G\bigr) (\gamma_k) }{\la F,\K^{N-1}G\ra}
    \ee
    against
    \be \label{eq:loop-lim}
        \int_{\R^{k-j+1}} \nu_{x -m_j,x'_j-m_j} \times
            \nu_{x'_j - m_{j+1}, x'_{j+1} - m_{j+1}} \times \cdots
         \times  \nu_{x'_k-m_k, y-m_k} \, \dd x'_j \cdots \dd x'_k.
    \ee
    Let $I_N(x,y;m_j,m_k)$ be the resulting integral. Note the one-to-one correspondence between the index set $\{1,\ldots,N\}$ and the finite lattice $\mathcal{L}_N = \{m_1,\ldots,m_N\}$, hence we may replace sums over $j$ and $k$ by sums over lattice points. The one-particle matrix is
    \be
        \rho_1^N(x;y) = \sum_{\stackrel{\ell,\ell'\in \mathcal{L}_N}{\ell\leq \ell'}} (-1)^{\rho (\ell'-\ell)} I_N(x,y;\ell,\ell').
    \ee
    Next, relabel the paths $\gamma_j,\ldots,\gamma_k$ as $\gamma_\ell,\gamma_{\ell+\lambda},\ldots,\gamma_{\ell'}$ with $\ell=m_j$ and $\ell'= m_k$.
    Let $j\to \infty$ and $N-k\to \infty$ in such a way that $\ell $ and $\ell'$ stay fixed. In this limit, the expression~\eqref{eq:qjk} converges  to
    \be \label{eq:qjk-lim}
        \Psi_0(-\gamma_\ell)   \frac{ K(\gamma_\ell,\gamma_{\ell+\lambda})}{z_0}  \cdots \frac{K(\gamma_{\ell'-\lambda},\gamma_{\ell'})}{z_0} \Psi_0(\gamma_{\ell'}).
    \ee
    uniformly on $E^{k-j}$. The measures $\nu_{xy}$ have total masses bounded by
    \be  \label{eq:nu-bound}
        \nu_{xy}(E) \leq C' \exp \Bigl( - C (x^2+y^2) \Bigr)
    \ee
    for suitable constants $C,C'>0$ and all $x,y\in R$ (see \eqref{eq:rd_deriv} and Proposition \ref{prop:gaussian}). Therefore the measure ~\eqref{eq:loop-lim} is a finite measure with total mass bounded by
    \be
        D^{\rho(\ell'-\ell)} \exp\Bigl( - C (x-\ell)^2 - C (y- \ell')^2 \Bigr)
    \ee
    for some $D>0$ and all $x,y\in \R$. We can exchange limits and integration: $I_N(x,y;\ell,\ell')$ converges to the integral $I(x,y;\ell,\ell')$ of the expression~\eqref{eq:qjk-lim} against the measure~\eqref{eq:loop-lim}.

 To check that we can also bring the limit \eqref{eq:thermolim} inside the sum
    $$
\lim \rho_1^N(x;y) = \lim \sum_{\stackrel{\ell,\ell'\in \mathcal{L}_N}{\ell\leq \ell'}} (-1)^{\rho (\ell'-\ell)} I_N(x,y;\ell,\ell'),
    $$
we bound $I_N(x,y; \ell,\ell')$ as follows. By Lemma~\ref{lem:krein} there is a $c>0$ such that for all $N,j,\gamma$,
    we have
    $\la F,\K^{N-1} G\ra \geq c z_0^{N-1}$,
    $ \frac{1}{z_0^{N-j}} \K^{N-j} G(\gamma) \leq c$, and $ \frac{1}{z_0^{j-1} }(\K^*)^{j-1} F(\gamma) \leq c$.  As a consequence,
    \be \label{eq:ibound}
    \begin{aligned}
         I_N(x,y;j,k)  \leq \Bigl(\frac{D C'}{z_0}\Bigr)^{\rho(\ell-\ell')}
            \exp\Bigl( - C (x-\ell)^2 - C (y- \ell')^2 \Bigr).
    \end{aligned}
    \ee
    The bound is independent of $N$ and its sum over $\ell,\ell'\in \mathcal{L}$ is finite, where
    $\mathcal{L} = \lambda/2+ \lambda \Z$.
 By dominated convergence, we see that we can exchange the sum and limit and obtain that
    \be \label{eq:one-point-function}
        \rho_1(x;y)= \lim \rho_1^N(x;y) = \sum_{\stackrel{\ell,\ell'\in \mathcal{L}}{\ell\leq \ell'}} (-1)^{\rho(\ell'-\ell)}   I(x,y;\ell,\ell')
    \ee
    The convergence is uniform on compact subsets of $\R\times \R$ because the sum over $\ell,\ell'\in \mathcal{L}$ of the last line in~\eqref{eq:ibound} is a locally bounded function of $x$ and $y$. This proves part (1) of Theorem~\ref{thm:corr-funct} for the one-particle matrix.
    The proof for the $n$-particle reduced density matrices is similar and therefore omitted; the roles of $j$ and $k$ (resp. $\ell$ and $\ell'$) are played, loosely speaking, by the smallest and largest index belonging to some open loops.

    In finite volume, the reduced density matrices $\rho_n^N$ are continuous functions of $\vect{x}$ and $\vect{y}$ because the integral kernel of $\exp(-\beta H_N)$ is continuous.
    The limits $\rho_n$, as  locally uniform limits of continuous functions, are also continuous.

    (ii) The invariance under shifts by multiples of $\lambda $ is immediate from the expressions~\eqref{eq:one-point-function}  and the covariance
    \be
        I(x-\lambda,y-\lambda;\ell,\ell') = I(x,y;\ell+ \lambda,\ell'+\lambda)
    \ee
    (and its $n$-particle analogues) inherited from the covariance of the measure~\eqref{eq:loop-lim}.

    In order to get to the smallest period, we apply Theorem~\ref{thm:symbreak2}. First we note that the diagonals ($\vect{x}=\vect{y}$) of the reduced density matrices are nothing else but the \emph{factorial moment densities} \cite[Chapter 5.4]{daley2007introduction}, also called \emph{product densities} or \emph{correlation functions},  of the measures $\nu_{a,b}^0$, considered as point processes. This statement survives in the thermodynamic limit. Thus for every interval $\cI\subset \R$, and all $n\in \N$,
\be
    \int_{\cI^n} \rho_n(\vect{x};\vect{x}) \dd \vect{x}
     = \sum_{k=0}^\infty k(k-1)\cdots(k-n+1)\, \nu_\R^0\Bigl(\{\text{there are exactly $k$ particles in $\cI$}\}\Bigr)
\ee
Let $N_{\cI}$ be the number of particles in the interval $\cI$. The previous equation shows that the set of functions $\rho_n(\vect{x};\vect{x})$, $n\in \N$, determine the moments of random variables $N_\cI$.  Because of Lemma~\ref{lem:ni-tight} below, $N_\cI$ satisfies Carleman's condition and the moments of $N_\cI$ determine the law of $N_\cI$ uniquely. Since the point process $\nu_\R^0$ in turn is uniquely determined by the law of the variables $N_\cI$, $\cI$ running over the intervals in $\R$, we see that the measure $\nu_\R^0$ is uniquely determined by the $\rho_n(\vect{x};\vect{x})$.

The same argument applies of course to the shifted measure $\nu_\R^0\circ \tau_\theta$ for $\theta \notin \lambda\Z$, which has factorial moment densities $\rho_n(\vect{x}-\theta;\vect{x}-\theta)$. The mutual singularity of the shifted measure to the original measure then implies that there must be an $n\in \N$ and an $\vect{x} \in \R^n$ such that
$\rho_n(\vect{x}-\theta;\vect{x}-\theta) \neq \rho_n(\vect{x};\vect{x})$.
\end{proof}

\begin{lemma} \label{lem:ni-tight}
    Let $N_{[x,y)}$ be the (random) number of particles in $[x,y)$.
    Then for suitable $\alpha,C>0$, all $x,y\in\R$ with $x<y$, and every $n\in \N$,
    \be
        \nu_\R^0\Bigl(  \bigl|N_{[x,y)}- \rho(y-x)\bigr| \geq n\Bigr) \leq  C \exp( -\alpha n^2).
    \ee
\end{lemma}

\begin{proof}
    In the proof of Theorem \ref{thm:symbreak2}, we showed that the point process is  $\nu_\R^0$ is the  law of $\{X_j\mid j \in \Z\}$ for a sequence  $(X_j)_{j\in \Z}$ of random variables such that $\cdots < X_j < X_{j+1} <\cdots$ and $\E X_j = (j-1/2)\lambda$. The $X_j$'s are the starting points of bridges.  We note that for suitable $C,\alpha >0$,
    \be
         \P\bigl( |X_j - (j-1/2)\lambda| \geq m \bigr) \leq C \exp( - \alpha m^2).
    \ee
    This follows because $X_j - (j-1/2)$ equals $\tilde \gamma_j(0)$, which has the law
    $\Psi_0(-\tilde \gamma)\Psi_0(\tilde \gamma) \nu(\dd\tilde  \gamma)$; $\Psi_0$ is bounded by Lemma~\ref{lem:krein}, and the law of $\gamma(0)$ under $\nu$ is Gaussian (see Appendix~\ref{app:ornstein}). For $x\in \R$, define the random variable
    \be \nn
        \mathcal{K}(x) := \mathrm{card} \{ j \in \Z \mid (j-1/2)\lambda <x,\  X_j \geq x \}
                 - \mathrm{\card} \{ j\in \Z \mid (j-1/2)\lambda \geq x,\ X_j <x\}.
    \ee
    $\mathcal{K}(x)$ is a particle excess number: it counts the number of particles that should be to the right of $x$ but are to the left, minus those that should be to the right but are to the left.
    The number of particles in the interval $[x,y)$ ($x<y$) equals
    \be \label{eq:kn}
        N_{[x,y)} = \card\{j\in \Z \mid (j-1/2)\lambda \in [x,y) \} + \mathcal{K}(x) - \mathcal{K}(y).
    \ee
    Lemma~\ref{lem:ni-tight} follows from estimates on $\mathcal{K}(x)$ and $\mathcal{K}(y)$.
    Consider first $\mathcal{K}(0)$.  Let $n\in \N_0$.
    By using that the $X_j$'s are ordered from left to right, we obtain
    \be
    \begin{aligned}
        \P( \mathcal{K}(0) \geq n ) & \leq \P \Bigl (  \mathrm{card} \{ j \in \Z \mid (j-1/2)\lambda <x,\  X_j \geq x \}  \geq n \Bigr) \\
                & \leq \P \bigl( X_{-n} \geq 0\bigr) = \P( Y_0 \geq n\lambda)
                \leq C \exp( - \alpha (n+\lambda/2)^2).
    \end{aligned}
    \ee
    A similar reasoning yields an estimate of $\P(\mathcal{K}(0)\leq - n)$ and of the deviation probabilities of $\mathcal{K}(x)$, $\mathcal{K}(y)$. Lemma~\ref{lem:ni-tight} then follows from Eq.~\eqref{eq:kn}.
\end{proof}

\appendix

\section{Non-colliding Ornstein-Uhlenbeck bridges}
\label{app:ornstein}
Let $E_0$ be the subset of $E$ consisting of all continuous paths on $[0,\beta]$ with the same starting and ending points.
Here we show that $\nu$, as defined in \eqref{eq:numeas}, is a probability measure on $E_0$ under which $(\gamma(t))_{0\leq t \leq \beta}$ is a Gaussian process:
\begin{prop}\label{prop:gaussian}
    The constants in \eqref{eq:rd_deriv} are calculated as
    \be \label{eq:hotrace}
        c(\beta,\rho) = \frac{1}{2 \sinh(\beta \sqrt{\rho/2})}= \frac{\exp( - \beta \sqrt{\rho/2})}{1 - \exp( - \beta \sqrt{2\rho})}.
    \ee
    Furthermore, for all $0\leq t_0<t_1<\cdots<t_r<\beta$, the vector $(\gamma(t_1),\ldots,\gamma(t_r))$ under the measure $\nu$ is Gaussian, and the variance of $\gamma(0)$  is
    \be
        \sigma(\beta,\rho)^2 = \Bigl[  2\sqrt{2\rho} \tanh\Bigl(\beta \sqrt{\frac{\rho}{2}}  \Bigr)\Bigr]^{-1}.
    \ee
\end{prop}
The result was essentially proven in~\cite[Eqs. (1.12)-(1.13)]{brascamp1975some}, we provide some more technical details.
\begin{proof}
     First we note that  $c(\beta,\rho) = \Tr \exp( - \beta A)$
    where $A$
    \be
        A = - \frac{1}{2} \frac{\dd^2}{\dd x^2} + \rho x^2
    \ee
    is the Hamiltonian of a harmonic oscillator. It is well-known that $A$, with a suitable domain, is a self-adjoint operator on $L^2(\R)$ \cite[Thm. X.28]{reed1980methods2}. The associated semi-group $\exp(-t A)$
is an integral operator with kernel
    \be \nn
        k_t(x,y) = \frac{(2\rho)^{1/4}}{\sqrt{2\pi \sinh(t\sqrt{2\rho} )}}\,
        \exp \left( -  \frac{\sqrt{2\rho} (x^2+y^2)}{2\tanh(t\sqrt{2\rho} )}
            + \frac{ \sqrt{2\rho} xy }{\sinh(t\sqrt{2\rho} )} \right)
    \ee
    (obtained from \emph{Mehler's formula} \cite[p. 55]{simon1979functional} by a change of variables).  In particular,
    \be
        k_\beta(x,x) =  \frac{(2\rho)^{1/4}}{\sqrt{2\pi \sinh(\beta\sqrt{2\rho} )}}
            \exp\Bigl( -\frac{x^2}{2\sigma(\beta,\rho)^2}\Bigr)
    \ee
    where
    \be \nn
        \frac{1}{2\sigma(\beta,\rho)^2} = \sqrt{2\rho} \tanh\Bigl(\beta \sqrt{\frac{\rho}{2}}  \Bigr).
    \ee
    The trace of $\exp(- \beta A)$ is
    \be \nn
        c(\beta,\rho) = \Tr \exp( - \beta A) = \frac{1}{2 \sinh(\beta \sqrt{\rho/2})}.
    \ee
    The law of $(\gamma(t_0),\ldots,\gamma(t_r))$ has a density proportional to
    \be \nn
        k_{t_1}(x_0,x_1)k_{t_2 - t_1} (x_1,x_2)\cdots k_{t_r- t_{r-1}}(x_{r-1},x_r)
            k_{\beta - t_r}(x_r,x_0),
    \ee
    which is a Gaussian. In particular, $\gamma(0)$ has a density proportional to $k_\beta(x,x)$, hence is a Gaussian with variance $\sigma(\beta,\rho)^2$.
\end{proof}

\begin{remark}
    Let us mention that $\nu$ is a mixture of Ornstein-Uhlenbeck bridges. This is because of the well-known relation between the Ornstein-Uhlenbeck process and the harmonic oscillator (see \cite[Thm 4.7]{simon1979functional}). We have, for example,
    \be
        - e^{ x^2\sqrt{\rho/2}} \Bigl( A - \sqrt{\frac{\rho}{2}}\mathrm{id}\Bigr)e^{-x^2\sqrt{\rho/2}} f = \bigl( \frac{1}{2} \frac{\dd^2}{\dd x^2} -\sqrt{2 \rho}\, x \frac{\dd}{\dd x} \bigr) f.
    \ee
    and we recognize the infinitesimal generator of the Markov process associated with the stochastic differential equation $\dd X_t = - \sqrt{2\rho}  X_t \dd t  + \dd B_t$.
\end{remark}

\paragraph{Acknowledgments} We thank M. Aizenman for posing this problem to us, and for pointing out references.  S. Jansen thanks the Weierstrass Institute, DFG Forschergruppe 718 ``Analysis and Stochastics in Complex
Physical Systems'', and ERC Advanced Grant 267356 VARIS of F. den Hollander for financial support. We also thank W. K{\"o}nig for kind hospitality  extended to P. Jung and for financial support  from the Weierstrass Institute where part of this work was done.

\bibliographystyle{alpha}
\bibliography{biblioMarch2013}

\end{document}